\documentclass[11pt,titlepage]{amsart}
\usepackage[ttscale=.875]{libertine}
\usepackage[margin=1in,paperwidth=8.5in,paperheight=11in]{geometry}

\usepackage[utf8]{inputenc}
\usepackage{amsthm,amssymb,amsmath,mathrsfs} 
\usepackage{amsaddr}

\usepackage{mathtools}
\usepackage{thmtools}
\usepackage{thm-restate}

\usepackage{microtype}
\usepackage{hyperref}
\usepackage{todonotes}

\DeclareFontFamily{U}{wncy}{}
\DeclareFontShape{U}{wncy}{m}{n}{<->wncyr10}{}
\DeclareSymbolFont{mcy}{U}{wncy}{m}{n}
\DeclareMathSymbol{\Sh}{\mathord}{mcy}{"58} 
\DeclareMathSymbol{\sha}{\mathord}{mcy}{"78}

\newcommand{\sa}{\mathcal{A}}

\newcommand{\momalg}[1]{{M}}

\newcommand{\convalg}[1]{{\mathcal{F}}}

\newcommand{\ann}{\operatorname{Ann}}
\newcommand{\tr}{\operatorname{tr}}
\newcommand{\skewalg}[1]{{S}}

\newcommand{\squext}[1]{{E^{\otimes 2}_n}}

\newcommand{\act}{\circ}

\DeclareMathOperator{\poly}{poly}

\DeclareMathOperator{\diff}{Diff}

\DeclareMathOperator{\diag}{diag}

\DeclareMathOperator{\lin}{span}

\DeclareMathOperator{\sgn}{sgn}
\newcommand\restr[2]{{
  \left.\kern-\nulldelimiterspace 
  #1 
  \vphantom{\big|} 
  \right|_{#2} 
  }}

\title{An Algorithmic Method of Partial Derivatives}
\author{Cornelius Brand}
\address{Charles University}
\email{cbrand@iuuk.mff.cuni.cz}

\author{Kevin Pratt}
\address{Carnegie Mellon University}
\email{kpratt@andrew.cmu.edu}

\declaretheorem[name=Theorem]{theorem}
\declaretheorem[name=Corollary]{cor}
\theoremstyle{plain}
\newtheorem{prop}{Proposition}
\newtheorem{lem}{Lemma}

\theoremstyle{definition}
\newtheorem{defn}{Definition}
\newtheorem{example}{Example}
\newtheorem{remark}{Remark}
\newtheorem{question}{Question}
\newtheorem{problem}{Problem}
\begin{document}

\begin{abstract}
We study the following problem and its applications: given a homogeneous degree-$d$ polynomial $g$ as an arithmetic circuit, and a $d \times d$ matrix $X$ whose entries are homogeneous linear polynomials, compute $g(\partial/\partial x_1, \ldots, \partial/\partial x_n) \det X$. By considering special cases of this problem we obtain faster parameterized algorithms for several problems, including the matroid $k$-parity and $k$-matroid intersection problems, faster \emph{deterministic} algorithms for testing if a linear space of matrices contains an invertible matrix (Edmonds's problem) and detecting $k$-internal outbranchings, and more. We also match the runtime of the fastest known deterministic algorithm for detecting subgraphs of bounded pathwidth, while using a new approach.

Our approach raises questions in algebraic complexity related to Waring rank and the exponent of matrix multiplication $\omega$. In particular, we study a new complexity measure on the space of homogeneous polynomials, namely the bilinear complexity of a polynomial's apolar algebra. Our algorithmic improvements are reflective of the fact that for the degree-$n$ determinant polynomial this quantity is at most $O(n 2^{\omega n})$, whereas all known upper bounds on the Waring rank of this polynomial exceed $n!$.
\end{abstract}
\maketitle
\newpage
\setcounter{page}{1}

\section{Introduction} \label{sec:intro}

Let $\mathcal{S}_d^n \coloneqq \mathbb{R}[x_1, \ldots, x_n]_d$ denote the vector space of homogeneous polynomials of degree $d$ in $n$ variables with real coefficients. We define the \emph{apolar inner product} $\langle \cdot , \cdot \rangle : \mathcal{S}_d^n \times \mathcal{S}_d^n \to \mathbb{R}$ via
\begin{equation}\label{apolar}
\langle f, g \rangle \coloneqq f(\frac{\partial}{\partial x_1}, \ldots, \frac{\partial}{\partial x_n}) g.
\end{equation}
This inner product (alternativley known as the \emph{Sylvester product}, the \emph{Bombieri inner product}, or the \emph{Fischer-Fock inner product}) originated in 19th century invariant theory \cite{sylvester1970principles} and has become a source of interest in computer science due to algorithmic applications. In a typical application, one first identifies some easy-to-evaluate generating polynomial $g$ whose coefficients encode solutions to a combinatorial problem. This information can then be recovered by computing $\langle f, g \rangle$ for a suitable choice of $f$. While this quantity is often $\#P$ hard to compute exactly (this follows from the coming example), in special cases it can be efficiently approximated. This approach has led to new algorithms for problems as disparate as approximating permanents and mixed discriminants \cite{gurvits2005complexity}, sampling from determinantal point processes \cite{anari2016monte}, Nash social welfare maximization \cite{anari2016nash}, and approximately counting subgraphs of bounded treewidth \cite{waring}.

For example, given a matrix $A \in \mathbb{R}^{n \times n}$, define $P_A \coloneqq \prod_{i=1}^n \sum_{j=1}^n A_{i,j} x_j$. Then 
\[\langle x_1x_2\cdots x_n, P_A \rangle = \sum_{\sigma \in \mathfrak{S}_n} \prod_{i=1}^n A_{i,\sigma(i)}\]
is the permanent of $A$.

As a second example, given a directed graph $G$ with $n$ vertices, let $A_G$ be the matrix with entry $(i,j)$ equal to the variable $x_i$ if there is an edge from vertex $v_i$ to vertex $v_j$, and zero otherwise. By the trace method,
\begin{equation*}\label{gencycles}
\mathrm{tr}(A_G^d) = \sum_{\substack{(v_{i_1},v_{i_2},\ldots,v_{i_d}) \in G, \\ v_{i_d} = v_{i_1}}} x_{i_1} \cdots x_{i_d} \in \mathcal{S}_d^n.
\end{equation*}
Now let $A \in \mathbb{R}^{d \times n}$ be a matrix any $d$ columns of which are linearly independent. Let $X = A \cdot \text{diag}(x_1, \ldots, x_n)\cdot A^\mathrm{T}$. By the Cauchy-Binet Theorem,
\[\det X =  \sum_{\substack{S \in \binom{[n]}{d}}} \det(A_S)^2 \prod_{i \in S}x_i.\]
(Here $A_S$ refers to the $d \times d$ submatrix of $A$ with columns indexed by the set $S$.) Since any $d$ columns in $A$ are linearly independent, $\det(A_S)^2 > 0$ for all $S \in \binom{[n]}{d}$. Then note that the result of differentiating $\mathrm{tr}(A_G^d)$ by $\det(A_S)^2 \prod_{i \in S}x_i$ is positive if there is a simple cycle on the vertices $\{v_i: i \in S\}$, and zero otherwise. It follows that $\langle \det X, \mathrm{tr}(A_G^d) \rangle > 0$ if and only if $G$ contains a simple cycle of length $d$.

Motivated by such examples, we consider the algorithmic task of computing \eqref{apolar} when $f$ is the determinant of a symbolic matrix (a matrix whose entries are homogeneous linear polynomials) and $g$ is given as an arithmetic circuit. This has applications to parameterized algorithms, yielding faster algorithms for the matroid $k$-packing and $k$-parity problems, the first deterministic $\poly(n)$-time algorithm for testing if a subspace of matrices of dimension $O(\log n)$ contains an invertible matrix, faster deterministic algorithms for detecting $k$-internal outbranchings, among others. Starting from the observation of the above example, we also give a deterministic $\varphi^{2d}\poly(n)<2.62^d\poly(n)$-time algorithm for detecting simple cycles of length $d$ in an $n$ vertex graph. Here $\varphi \coloneqq \frac{1+\sqrt{5}}{2}$ is the golden ratio. This brushes up against the fastest known algorithm for this problem which has runtime $2.55^d \poly(n)$ \cite{tsur}. Our algorithm also generalizes to detecting subgraphs of bounded pathwidth, unexpectedly matching the runtime of the fastest known algorithm for this problem \cite{fomin}, while using a new, very mechanical, approach.

Our algorithms for computing special cases of \eqref{apolar} turn out to be equivalent to algorithms for performing arithmetic in a certain algebra $\sa_f$ associated to $f$, namely the \emph{apolar algebra} of $f$. Apolar algebras (also known as Artinian Gorenstein algebras) have been studied extensively since the work of F.S.~Macaulay in 1916 \cite{macaulay1994algebraic} and are ubiquitous in algebraic combinatorics; see e.g.~\cite{adiprasito2018hodge}. As a first step towards extending our approach, we then study $\mathbf{R}(\sa_f)$, the bilinear complexity of the apolar algebra of a polynomial $f$. This gives upper bounds on the number of non-scalar multiplications needed to compute \eqref{apolar} in the white-box setting (Proposition \ref{multbd}). We will show in Example \ref{fsc} that in fact previous methods in exact algorithms (specifically, those for subset convolution) necessarily made use of upper bounds on this quantity.

To obtain further algorithmic improvements, we pose the following algebraic question:

\begin{question}
Let $\mathcal{T}_{d,n}$ be the set of all $f \in \mathcal{S}_d^n$ such that $f = \sum_{S \in \binom{[n]}{d}} c_S \prod_{i \in S} x_i$, where $c_S > 0$ for all $S$.  What is $B(d,n) \coloneqq \min (\dim \diff(f) :f \in \mathcal{T}_{d,n})$?
Here $\diff(f)$ denotes the vector space spanned by the partial derivatives of all orders of $f$.
\end{question}

This question was asked in \cite[Question 73]{waring}, but it was not known that an answer would have algorithmic implications.  Our algorithms make use of the upper bound $B(d,n) < \varphi^{2d}$, obtained by taking $f$ to be the determinant of a symbolic Hankel matrix. We remark that it is not hard to show that $B(d,n) \ge 2^d$. A proof of this fact is sketched as follows: first, observe that $\dim \diff (f)$ does not increase under zeroing out variables. Hence for any $f \in \mathcal{T}_{d,n}$, $\dim \diff(f) \ge \dim \diff(c \cdot x_1x_2 \cdots x_d)$ for some nonzero constant $c$. As $\diff(c \cdot x_1x_2 \cdots x_d)$ is spanned by the collection of products of subsets of the variables $x_1, \ldots, x_d$, the claim follows.
\subsection{Previous approaches to computing the inner product \eqref{apolar}}
One special case of \eqref{apolar} that has been the source of several recent breakthroughs is when $f$ and $g$ are \emph{real stable} polynomials with nonnegative coefficients; see e.g.~\cite{gurvits2008van,anari2018log}. In this case $\langle f, g \rangle$ can be approximated (up to a factor of $e^{d+\varepsilon}$) in polynomial time by a reformulation as a convex program \cite[Theorem 1.2]{anari2017generalization}. For the cases we consider, however, $f$ and $g$ will not be real stable.

Another approach is based on \emph{Waring rank} upper bounds \cite{barvinok1996two,gurvits2006hyperbolic,glynn2013permanent,waring}. The Waring rank of $f \in \mathcal{S}_d^n$, denoted $\mathbf{R}_S(f)$, is defined as the minimum $r$ such that $f = \sum_{i=1}^r  c_i\ell_i^d$ for linear forms $\ell_1, \ldots, \ell_r \in \mathcal{S}_1^n$ and scalars $c_1, \ldots, c_r$. For example, the identity
\[x_1x_2x_3 = \frac{1}{24}\left [(x_1+x_2+x_3)^3-(x_1+x_2-x_3)^3-(x_1-x_2+x_3)^3-(-x_1+x_2+x_3)^3\right ]\]
shows that $\mathbf{R}_S(x_1x_2x_3) \le 4$. Waring rank has been studied in invariant theory and algebraic geometry since the 1850's \cite[Introduction]{ik} and has gained recent attention for its applications to algebraic complexity \cite{burgisser2019no, chiantini2018polynomials}. Its relevance to \eqref{apolar} is due to the following fact, which can be verified by a direct calculation: if $f = \sum_{i=1}^r c_i(a_{i,1}x_1 + \cdots + a_{i,n}x_i)^d$, then for all $g \in \mathcal{S}_d^n$,
\[\langle f, g \rangle = d! \sum_{i=1}^r c_ig(a_{i,1}, \ldots, a_{i,n}).\] Hence upper bounds on $\mathbf{R}_S(f)$ yield algorithms for computing $\langle f, g \rangle$. Furthermore, it was shown in \cite[Theorem 6]{waring} that with only evaluation access to $g$, $\mathbf{R}_S(f)$ queries are \emph{required} to compute this inner product. Unfortunately, $\mathbf{R}_S(f)$ is usually prohibitively large; for instance, the Waring rank of almost all $f \in \mathcal{S}_d^n$ is at least$\lceil \binom{n+d-1}{d}/n \rceil$ \cite[Section 3.2]{landsberg2012tensors}.

In \cite{waring} this difficulty was overcome by studying relaxations of Waring rank. For instance, while the elementary symmetric polynomial $e_{n,d}$ is known to have Waring rank roughly $n^{d/2}$ \cite{lee2016power}, for all $\varepsilon > 0$ there exists a polynomial $f_\varepsilon \in \mathcal{S}_d^n$ with Waring rank only $O(\frac{4.075^d \log n}{\varepsilon^2})$ that $\varepsilon$-\emph{approximates} $e_{n,d}$, in the sense that for all $g \in \mathcal{S}_d^n$,
\[(1-\varepsilon)\langle e_{n,d}, g \rangle \le  \langle  f_\varepsilon, g \rangle  \le (1+\varepsilon)\langle e_{n,d}, g \rangle.\]
This fact lends itself to parameterized algorithms for problems such as approximately counting simple cycles. We remark that this discrepancy between $\mathbf{R}_S(e_{n,d})$ and the Waring rank of polynomials ``close'' to $e_{n,d}$ can be understood from a parameterized algorithms perspective as reflecting the fact that exactly counting cycles of a given length is $\#W[1]$ hard \cite{flum2004parameterized}, whereas the problem of approximately counting cycles has been known to admit parameterized algorithms since \cite{arvind2002approximation}.

\subsection{Our approach}
In contrast to previous approaches, we consider the white-box setting where $g$ is given as an arithmetic circuit $C$. For us, $f$ will always be the determinant of a symbolic matrix $X$ that is given as input. Our algorithms work by inductively evaluating $C$, computing at each gate the result of differentiating $\det X$ by the polynomial computed by $C$ at that gate\footnote{By ``differentiating $f$ by $g$'' we mean applying the differential operator $g(\partial/\partial x_1, \ldots, \partial/\partial x_n)$ to $f$.}. At the end of the algorithm, the output gate of $C$ will therefore contain $\langle g,f \rangle = \langle f, g \rangle$. Here we make use of the fact that $\langle \cdot , \cdot \rangle$ is symmetric, so one can either think about differentiating $f$ by $g$ or vice versa.

The key to our approach is that for a symbolic $d \times d$ matrix $X$, the vector space of partial derivatives of $\det X$ has dimension at most $4^d$, and in some important cases this bound can be improved to $\varphi^{2d}$. So while one might na\"{i}vely represent an element in this space as a linear combination of $\binom{n+d}{d}$ monomials, doing so generally includes a significant amount of unnecessary information. Instead, we represent elements in this space as linear combinations of minors (determinants of submatrices) of $X$, which are specified by pairs of increasing sequences.

We will start by giving in our Theorem \ref{gendiff} an algorithm for the special but important case when $g$ is computed by a \emph{skew} circuit, meaning one of the two operands to each multiplication gate is a variable or a scalar:

\begin{restatable*}{theorem}{gendiff}
\label{gendiff}
Let $C$ be a skew arithmetic circuit computing $g \in \mathcal{S}_d^n$, and let $X = (\ell_{i,j})_{i,j \in [d]}$ be a symbolic matrix with entries in $\mathcal{S}_1^n$. Then we can compute $\langle \det X, g \rangle$ with $4^d |C| \poly(d)$ arithmetic operations.
\end{restatable*}
Our algorithm for Theorem \ref{gendiff} only uses linear algebra and basic properties of differentials.

Of particular interest will be the case of Theorem \ref{gendiff} when $X$ is a Hankel matrix, meaning that $X_{i,j} = X_{i+k,j-k}$ for all $k = 0, \ldots, j-i$. For example, the generic $3 \times 3$ Hankel matrix is
\[\begin{bmatrix}
x_1 & x_2 & x_3\\ 
x_2 & x_3 & x_4\\ 
x_3 & x_4 & x_5
\end{bmatrix}.\]
We show the following improvement in this special case:
\begin{restatable*}{theorem}{hankeldiff}
\label{hankeldiff}
Let $C$ be a skew arithmetic circuit computing $g \in \mathcal{S}_d^n$, and let $X = (\ell_{i,j})_{i,j \in [d]}$ be a symbolic Hankel matrix with entries in $\mathcal{S}_1^n$. Then we can compute $\langle \det X, g \rangle$ with $\varphi^{2d}\poly(d)|C|$ arithmetic operations. Here $\varphi \coloneqq \frac{1+\sqrt{5}}{2}$ is the golden ratio.
\end{restatable*}
The improvement in Theorem \ref{hankeldiff} over Theorem \ref{gendiff} is facilitated by the fact that the space of partial derivatives of the determinant has dimension about $4^d$, whereas the space of partial derivatives of the determinant of a Hankel matrix has dimension less than $\varphi^{2d}$. We also make use of use of linear relations in the space of minors of a Hankel matrix originally studied in commutative algebra \cite{conca}.

\subsection{Applications to parameterized algorithms}
Theorem \ref{gendiff} yields faster algorithms for the $k$-matroid intersection and matroid $k$-parity problems. These are the following problems:

\begin{problem}[Matroid $k$-Parity]\label{p1}
Suppose we are given a matrix $B \in \mathbb{Q}^{km \times kn}$ representing a matroid $M$ with groundset $[kn]$, and a partition $\pi$ of $[kn]$ into parts of size $k$. Decide if the union of any $m$ parts in $\pi$ are independent in $M$.
\end{problem}

\begin{problem}[$k$-Matroid Intersection]\label{p2}
Suppose we are given matrices $B_1, \ldots, B_k \in \mathbb{Q}^{m \times n}$ representing matroids $M_1, \ldots ,M_k$ with the common groundset $[n]$. Decide if $M_1, \ldots, M_k$ share a common base. 
\end{problem}

We show in Theorems \ref{matroid1} and \ref{kmi} that these can be solved in time $4^{km}\poly(N)$, where $N$ denotes the size of the input. When $k=2$ these are the classic matroid parity and intersection problems and can be solved in polynomial time, but for $k > 2$ they are NP-hard. The first algorithms for general $k$ faster than na\"{i}ve enumeration were given by Barvinok in \cite{barvinok1995new}, and had runtimes $(km)^{2k+1}4^{k m} \poly(N)$ and $(km)^{2k}4^{k^2m} \poly(N)$, respectively. A parameterized algorithm for Problem \ref{p1} was also given by Marx in \cite{marx} where it was used to give fixed-parameter tractable algorithms for several other problems, including Problem \ref{p2}. The fastest algorithms prior to our work were due to Fomin et al.~\cite{fomin} and had runtime $2^{km \omega} \poly(N)$, where $\omega < 2.373$ is the exponent of matrix multiplication \cite{le2014powers}.

By combining Theorem \ref{gendiff} with a known construction of the determinant as a skew circuit \cite{mahajan1997combinatorial}, we obtain a faster deterministic algorithm for the following problem:
\begin{problem}[SING]
Given matrices $A_1, \ldots, A_n \in \mathbb{Q}^{d \times d}$, decide if their span contains an invertible matrix. Equivalently, decide if $\det \sum_{i=1}^n x_i A_i \not \equiv 0$.
\end{problem}
We show that SING can be solved in $4^d \poly(N)$ time in our Corollary \ref{pit}. In particular, this establishes that $\text{SING} \in \mathcal{P}$ for subspaces of matrices of logarithmic dimension. The fastest previous algorithm, given by Gurvits in \cite{gurvits}, had runtime $2^d d! \poly(N)$ and made use of an upper bound of $2^d d!$ on $\mathbf{R}_S(\det_d)$. This problem was originally studied by Edmonds for its application to matching problems \cite{edmonds1967systems}. While it is known to admit a simple randomized polynomial time algorithm as was first observed by Lov\'asz \cite{Lovasz1979determinants}, a \emph{deterministic} polynomial time algorithm would imply circuit lower bounds that seem far beyond current reach \cite{kabanets2004derandomizing}. As a result, variants of SING have attracted attention, leading to a recent breakthrough in the non-commutative setting \cite{garg2019operator}.

Theorem \ref{hankeldiff} yields the following applications:
\begin{restatable*}{cor}{apps}
\label{apps}
The following problems admit deterministic algorithms running in time $\varphi^{2d}\poly(n)$:
\begin{enumerate}
\item Deciding whether a given directed $n$-vertex graph has a directed spanning tree with at least $d$ non-leaf vertices,
\item Deciding whether a given edge-colored, directed $n$-vertex graph has a directed spanning tree containing at least $d$ colors,
\item Deciding whether a given planar, edge-colored, directed $n$-vertex graph has a perfect matching containing at least $d$ colors.
\end{enumerate}
\end{restatable*}

The previous fastest algorithms for these problems had runtimes $3.19^d\poly(n),$ $4^d\poly(n)$, and $4^d\poly(n)$, respectively \cite{esa}. This built upon work of Gutin et al.~\cite{gutin} Problem (1) is the best studied among these, with  \cite[Table 1]{gutin} listing eleven articles on this problem in the last fourteen years. It is noteworthy that our improvements do not rely on any problem-specific adaptations.

Theorem \ref{hankeldiff} also yields a $\varphi^{2d}\poly(n)$-time deterministic algorithm for detecting simple cycles of length $d$ in an $n$ vertex directed graph (and paths, and more generally subgraphs of bounded treewidth). While it is known that simple cycles of length $d$ can be detected in randomized time $2^d \poly(n)$ \cite{williams09} ($1.66^d \poly(n)$ for undirected graphs \cite{bjoerklund}), it is a major open problem to achieve the same runtime deterministically. Our algorithm brushes up against the fastest known deterministic algorithm for this problem which has runtime $2.55^d \poly(n)$ \cite{tsur}, and unexpectedly matches the runtime of a previous algorithm \cite{fomin} while using a different (shorter) approach. Our approach differs from those of previous algorithms which have been based on paradigms such as \emph{color coding}, \emph{divide and color}, and \emph{representative families} \cite[Chapter 5]{fominbook}. Whereas these methods make use of explicit constructions of pseudorandom objects such as perfect hash families, universal sets, and representative sets, our algorithm makes use of algebraic-combinatorial identities. This approach was foreshadowed in \cite[Theorem 2]{extensor}. It is important to note that our algorithm only works for unweighted graphs (or weighted graphs with integer weights bounded by $\poly(n)$), while several previous algorithms work for weighted graphs. The algorithm of \cite{fomin} also extends more generally to detect subgraphs of bounded treewidth. 

\subsection{Algebraic considerations}
As we note in Remark \ref{rmkq}, our Theorems \ref{gendiff} and \ref{hankeldiff} can be viewed as algorithms for multiplication by degree-1 elements in the apolar algebras of the determinant and the generic Hankel determinant, respectively. Algorithms for \emph{general} multiplication in these algebras, however, would have applications to problems such as detecting subgraphs of bounded \emph{treewidth} (rather than just pathwidth). As a first step towards this, we consider the quantity $\mathbf{R}(\sa_f)$, the bilinear complexity of $\sa_f$. This is at most twice the number of non-scalar multiplications needed to multiply two elements in $\sa_f$, and we show in Proposition \ref{multbd} how it bounds the number of multiplications that can be used to compute \eqref{apolar}. We note in our Theorem \ref{waringrelax} that $\mathbf{R}(\sa_f)$ is, up to a linear factor, a lower bound on the Waring rank of $f$. 

We show that for the degree-$n$ determinant polynomial $\det_n$, $\mathbf{R}(\sa_{\det_n}) \le O(n 2^{\omega n})$ where $\omega < 2.373$ is the exponent of matrix multiplication \cite{le2014powers}. Our upper bound on $\mathbf{R}(\sa_{\det_n})$ follows by realizing the apolar algebra of the determinant as a limit of a tensor product of Clifford algebras, which are classically known to be isomorphic to matrix algebras. We point out that if our upper bound on $\mathbf{R}(\sa_{\det_n})$ is optimal, by Theorem \ref{waringrelax} we would have that $\mathbf{R}_S(\det_n) \ge \Omega(2^{\omega n})$. For reference, the best known lower bounds on $\mathbf{R}_S(\det_n)$ are roughly $4^n$. Therefore if known lower bounds on $\mathbf{R}_S(\det_n)$ \emph{and} our upper bound on $\mathbf{R}(\sa_{\det_n})$ are roughly optimal, then $\omega = 2$.  The Waring rank of the determinant and $\omega$ have been studied primarily in different contexts (see \cite{alper2019syzygies,shafiei1,boij2020bound,derksen2015lower} for work on the former), and have only started to be related \cite{chiantini2018polynomials}.

\subsection{Paper outline}
In the next section we prove Theorems \ref{gendiff} and \ref{hankeldiff}. In Section 3 we give our applications. These follow quickly from Theorems \ref{gendiff} and \ref{hankeldiff}, using little more than Cauchy-Binet. In Section 4 we then define the apolar algebra of a polynomial, and briefly discuss tensor rank and bilinear complexity. Using these we then define $\mathbf{R}(\sa_f)$. We show in Example \ref{fsc} how the fast subset convolution algorithm of \cite{bjorklund2007fourier} is equivalent to an algorithm for multiplication in $\sa_{x_1x_2 \cdots x_n}$, and how one can deduce from known tensor rank upper bounds improved upper bounds on the number of non-scalar multiplications needed to compute subset convolution. Finally we give our upper bound on $\mathbf{R}(\sa_{\det_n})$.

\section{Computing the apolar inner product for skew circuits}
We start by giving an algorithm for computing \eqref{apolar} in the case that $g$ is the determinant of a symbolic matrix and $f$ is computed by a skew arithmetic circuit $C$.  This is a warmup for the special case when $g$ is the determinant of a symbolic Hankel matrix.

We fix the following notation for the rest of the paper. We denote by $|C|$ the total number of gates in the circuit $C$. Let $\mathbb{N}^k_d$ be the set of $k$-tuples with elements in $[d]$, and let $I(d,k) \subseteq \mathbb{N}^k_d$ be the set of strictly increasing sequences of length $k$ with elements in $[d]$; when $k=0$ we include the empty sequence in this set. Given a $d \times d$ matrix $X$ and tuples $\alpha,\beta \in I(d,k)$, we denote by $X[\alpha| \beta]$ the minor (determinant of a submatrix) of $X$ with rows indexed by $\alpha$ and columns indexed by $\beta$. We declare the ``empty minor'' $X[\ |\  ]$ to equal one. We use the convention of writing $\alpha_1, \ldots, \widehat{\alpha}_i, \ldots, \alpha_k$ to denote the sequence $\alpha_1, \ldots, \alpha_{i-1}, \alpha_{i+1}, \ldots, \alpha_k$ obtained from $\alpha$ by omitting $\alpha_i$. We call a monomial $x_1^{a_1} \cdots x_n^{a_n}$ \emph{square-free} if $a_i \in \{0,1\}$ for all $i$.

 For $f \in \mathcal{S}_d^n$, $\diff(f)$ denotes the vector space spanned by the partial derivatives of $f$ of all orders (this includes $f$ itself). For example, $\diff(x_1x_2)$ is the vector space spanned by $x_1x_2,x_1,x_2,$ and $1$. The next observation is a simple bound on this quantity for determinants of symbolic matrices, and has been essentially observed several times previously (e.g.~\cite[Lemma 1.3]{shafiei1}).

\begin{prop}\label{det_gen_bd}
Let $X = (\ell_{i,j})_{i,j \in [d]}$ be a symbolic matrix with entries in $\mathcal{S}_1^n$. Then $\diff(\det X)$ is contained in the space of minors of $X$. Hence 
\[\dim \diff(\det X) \le \sum_{i=0}^d \binom{d}{i}^2 = \binom{2d}{d} < 4^d.\] 
\end{prop}
\begin{proof}
Let $\mathfrak{S}_d$ denote the symmetric group on $d$ elements. By the Leibniz formula for the determinant and the product rule, for any $l \in [n]$,
\begin{align*}
\frac{\partial \det X}{\partial x_l} &= \sum_{\sigma \in \mathfrak{S}_d} \sgn(\sigma) \sum_{i=1}^d \frac{\partial \ell_{i,\sigma(i)}}{\partial x_l} \prod_{j \neq i} \ell_{j, {\sigma(j)}} =\sum_{1 \le i,j \le d} \frac{\partial \ell_{i, j}}{\partial x_l} \sum_{\sigma \in \mathfrak{S}_d, \sigma(i) = j} \sgn(\sigma) \prod_{m \neq i} \ell_{m, {\sigma(m)}}\\
 &= \sum_{1 \le i,j \le d}  (-1)^{i+j} \frac{\partial \ell_{i, j}}{\partial x_l} X[1, \ldots, \widehat{i}, \ldots, d | 1, \ldots, \widehat{j}, \ldots, d].
\end{align*}
Note that $\frac{\partial \ell_{i, j}}{\partial x_l}$ is just a scalar. To see the last equality, consider the martix $X^{(ij)}$ obtained by setting the $(i,j)$th entry of $X$ to $1$, and all other entries in the $i$th row of $X$ to 0. Then $\det X^{(ij)} = \sum_{\sigma \in \mathfrak{S}_d, \sigma(i)=j} \sgn(\sigma) \prod_{m \neq i} \ell_{m,\sigma(m)}$, but at the same time by Laplace expansion along the $i$th row of $X^{(ij)}$, $\det X^{(ij)}  = (-1)^{i+j}X[1, \ldots, \widehat{i}, \ldots, d| 1, \ldots, \widehat{j}, \ldots, d]$.

This shows that the space of order-1 partial derivatives of $\det X$ is contained in the span of the degree-$(d-1)$ minors of $X$. That $\diff (\det X)$ is contained in the space of minors of $X$ follows by repeated application of this fact. Furthermore, since square $k \times k$ submatrices of $X$ can be identified by pairs of elements in $I(d,k)$ (their row and column indices), the vector space spanned by all minors of $X$ has dimension at most $\sum_{k=0}^d |I(d,k)|^2 = \sum_{k=0}^d \binom{d}{k}^2 = \binom{2d}{d}$. 
\end{proof}


\begin{lem}\label{linear_diff_gen}
Given as input a symbolic matrix $X = (\ell_{i,j})_{i,j \in [d]}$ with entries in $\mathcal{S}_1^n$, a linear combination $P$ of minors of $X$, and $l \in [n]$, we can compute a representation for $\frac{\partial P}{\partial x_l}$ as a linear combination of minors of $X$ with $4^d \poly(d)$ arithmetic operations.
\end{lem}
\begin{proof}
Let $P = \sum_{k=0}^d \sum_{\alpha, \beta \in I(d,k)} c_{\alpha,\beta}X[\alpha | \beta]$ and let $a_{i,j}^{(l)}$ be the coefficient of $x_l$ in $\ell_{i,j}$ (so the input consists of $l$ and the vectors $(c_{\alpha,\beta}) \in \mathbb{R}^{\binom{2d}{d}}$,  $(a_{i,j}^{(k)}) \in \mathbb{R}^{d^2n}$). Then by the same considerations as in the proof of  Proposition \ref{det_gen_bd},
\[\frac{\partial P}{\partial x_l} = \sum_{k=1}^d \sum_{\alpha, \beta \in I(d,k)} \sum_{1 \le i,j \le k} c_{\alpha,\beta} (-1)^{i+j} a_{i,j}^{(l)} X[\alpha_1, \ldots, \widehat{\alpha}_i, \ldots, \alpha_k | \beta_1, \ldots, \widehat{\beta}_j, \ldots, \beta_k].\]
Note that for $\alpha, \beta \in I(d,k)$, the coefficient of $X[\alpha | \beta]$ in the above equals
\[\sum_{1 \le i,j \le k} \sum_{\substack{\alpha', \beta' \in I(d,k+1) \\ \alpha = \alpha'_1, \ldots, \widehat{\alpha}'_i, \ldots, \alpha'_{k+1} \\ \beta = \beta'_1, \ldots, \widehat{\beta}'_j, \ldots, \beta'_{k+1}}} (-1)^{i+j} a_{i,j}^{(l)}c_{\alpha', \beta'}.\]
The numbers of pairs of sequences $\alpha', \beta'$ considered by the inner sum is na\"{i}vely bounded by $d^4$ (there are $d$ positions in $\alpha$ where we could try to insert a number in $[d]$ into to get an increasing sequence, and similarly for $\beta$), and hence the coefficient of each minor can be computed with $O(d^6)$ arithmetic operations. Since there are $\binom{2d}{d}$ minors, all coefficients can be computed with the stated number of operations.
\end{proof}

\gendiff

\begin{proof}
Say that gate $v$ in $C$ computes the polynomial $C_v$. We will compute the inner product \eqref{apolar} inductively: at gate $v$ we will compute and store $C_v^\partial$, a representation for $C_v(\frac{\partial}{\partial x_1}, \ldots, \frac{\partial}{\partial x_n})\det A$ as a linear combination of minors of $X$. $C_v^\partial$ will be stored as a vector of length $\binom{2d}{d}$ indexed by pairs of row and column sets. At the end of the algorithm we will have computed $f(\frac{\partial}{\partial x_1}, \ldots, \frac{\partial}{\partial x_n}) \det X = \langle f, \det X \rangle$ at the output gate. 

We start by computing and storing $\frac{\partial}{\partial x_l} \det X$ at input gate $x_l$, which by Lemma \ref{linear_diff_gen} can be done in $4^d \poly(d)$ time. Now suppose that gate $v$ takes input from gates $v'$ and $v''$, and that we have already computed $C_{v'}^\partial$ and $C_{v''}^\partial$. To compute $C_v^\partial$, there are two cases to consider:
\begin{enumerate}
\item $C_v = x_i \cdot C_{v'}$. Then $C_v^\partial = \frac{\partial}{\partial x_i} C_{v'}(\frac{\partial}{\partial x_1}, \ldots, \frac{\partial}{\partial x_n})\det A = \frac{\partial}{\partial x_i}  C_{v'}^\partial$. Using Lemma \ref{linear_diff_gen} this can be computed with $4^d \poly(d)$ operations.
\item $C_v = C_{v'} + C_{v''}$. Since differentiation is linear, $C_v^\partial = C_{v'}^\partial + C_{v''}^\partial$. Since $C_{v'}^\partial$ and $C_{v''}^\partial$ are vectors of length $\binom{2d}{d}$, it takes $\binom{2d}{d}$ operations to add them.
\end{enumerate}
Hence at each gate we use at most $4^d \poly(d)$ arithmetic operations, for a total of $4^d \poly(d)|C|$.
\end{proof}
We now show how Theorem \ref{gendiff} can be applied to obtain a deterministic algorithm for detecting simple cycles in graphs. This is not competitive, but it motivates our following improvement.

\begin{prop}\label{test_c}
Let $G$ be a graph on $n$ vertices. We can decide in $4^d \poly(n)$ time if $G$ contains a simple cycle of length $d$.
\end{prop}
\begin{proof}
Let $V \in \mathbb{Q}^{d \times n}$ be the Vandermonde matrix with $V_{i,j}=j^i$. Let $X = V \cdot \text{diag}(x_1, \ldots, x_n)\cdot V^\mathrm{T}$. By the Cauchy-Binet Theorem,
\[\det X =  \sum_{\substack{\alpha \in I(n,d)}} V[1, \ldots, d| \alpha]^2 \prod_{i \in S}x_i.\]
Since any $d$ columns in $V$ are linearly independent, $V[1, \ldots, d| \alpha]^2 > 0$ for all $\alpha \in I(n,d)$. Furthermore, observe that $\tr(A_G^d)$ has nonnegative coefficients and contains a square-free monomial if and only if $G$ contains a simple cycle of length $d$. It follows that $\langle \det A, \tr(A_G^d) \rangle \neq 0$ if and only if $G$ contains such a cycle. In addition, $\tr(A_G^d)$ can be na\"ively computed by a skew circuit of size $O(dn^3)$. The theorem follows by applying Theorem \ref{gendiff}, noting that we only perform arithmetic with $\poly(n)$-bit integers.
\end{proof}
Note that the $(i,j)$th entry in the matrix $X$ in the proof of  Proposition \ref{test_c} equals $\sum_{k=1}^n k^{i+j} x_k$, and therefore $X$ is Hankel. We now show how this additional structure can be exploited.

 Fix linear forms $\ell_1, \ldots, \ell_{2d-1} \in \mathcal{S}_1^n$, and let $C_d$ be the symbolic matrix
\begin{align} \label{eq:arrangement}
\begin{bmatrix}
 \ell_1 & \ell_2 & \ell_3 & \cdots &  \cdots  &  \cdots  & \ell_{2d-2} & \ell_{2d-1}\\ 
\ell_2 & \ell_3 & \reflectbox{$\ddots$}  & \reflectbox{$\ddots$}  & \reflectbox{$\ddots$}  & \reflectbox{$\ddots$} & \reflectbox{$\ddots$}  & 0\\ 
\ell_3 & \reflectbox{$\ddots$}  & \reflectbox{$\ddots$}  &\reflectbox{$\ddots$}  & \reflectbox{$\ddots$} &\reflectbox{$\ddots$} & \reflectbox{$\ddots$}  & 0\\ 
\vdots & \reflectbox{$\ddots$}  & \reflectbox{$\ddots$}  & \reflectbox{$\ddots$} & \reflectbox{$\ddots$} & \reflectbox{$\ddots$}  & \reflectbox{$\ddots$}  & \vdots \\ 
\vdots & \reflectbox{$\ddots$}  & \reflectbox{$\ddots$} & \reflectbox{$\ddots$} & \reflectbox{$\ddots$} & \reflectbox{$\ddots$} & \reflectbox{$\ddots$} & \vdots \\ 
\vdots & \reflectbox{$\ddots$} & \reflectbox{$\ddots$} & \reflectbox{$\ddots$}  & \reflectbox{$\ddots$} &  \reflectbox{$\ddots$} & \reflectbox{$\ddots$} & \vdots\\ 
\ell_{2d-2} & \reflectbox{$\ddots$}  & \reflectbox{$\ddots$}  & \reflectbox{$\ddots$}  & \reflectbox{$\ddots$}  & \reflectbox{$\ddots$}  & \reflectbox{$\ddots$}  & \vdots \\ 
\ell_{2d-1} & 0 & 0  &  \cdots&\cdots  &  \cdots&  \cdots & 0
\end{bmatrix}.
\end{align}
The minors of the form $C_d[1,2,\ldots,k | b_1, \ldots, b_k]$, where $k \le d$ and $b_k \le 2d-k$, are called \emph{maximal}. For brevity we denote such a minor by $C_d[b_1, \ldots, b_k]$. Let $H_d$ be the submatrix of $C_d$ with row and column subscripts $1, \ldots, d$. It is readily seen that $H_d$ is a Hankel matrix.

\begin{prop}\label{cat_bd}
$\diff(\det H_d)$ is contained in the space of maximal minors of $C_d$. Furthermore, the number of maximal minors of $C_d$ is at most $\varphi^{2d}$.
\end{prop}
\begin{proof}
The maximal minors of $C_d$ span the space of minors of $H_d$ by Corollary 2.2(c) of \cite{conca}. Hence by Proposition \ref{det_gen_bd}, they span the space of partial derivatives of $\det H_d$. The second claim follows by noting that the number of maximal minors of degree $k$ equals $|I(2d-k,k)| = \binom{2d-k}{k}$. Hence the total number of maximal minors equals $\sum_{k=0}^d \binom{2d-k}{k} < \varphi^{2d}$. In the last step we used the facts that the $d$th Fibonacci number satisfies $F_d = \sum_{k=0}^{\lfloor \frac{d-1}{2} \rfloor} \binom{d+k-1}{k}$, and that $F_d \le \varphi^{d-1}$.
\end{proof}

\begin{lem}\label{cat_ldiff}
Given as input a linear combination $P$ of maximal minors of $C_d$ and $l \in [n]$, we can compute a representation for $\frac{\partial P}{\partial x_l}$ as a linear combination of maximal minors of $C_d$ with $\varphi^{2d} \poly(d)$ arithmetic operations.
\end{lem}
\begin{proof}
For brevity we will write $[\alpha]$ for the minor $C_d[\alpha]$. Let $P = \sum_{k=0}^d \sum_{\beta \in I(2d-k,k)} c_{\beta}[\beta]$, and say that the coefficient of $x_l$ in $(C_d)_{i,j}$ is $a_{i,j}^{(l)}$. As in Lemma \ref{linear_diff_gen},
\[\frac{\partial P}{\partial x_l} = \sum_{k=1}^d \sum_{\beta \in I(2d-k,k)} c_{\beta} \sum_{1 \le i,j \le k}  (-1)^{i+\beta_j} a_{i,\beta_j}^{(l)} [1, \ldots, \widehat{i}, \ldots, k | \beta_1, \ldots, \widehat{\beta}_j, \ldots, \beta_k].\]
Note that the only minors with nonzero coefficient in this expression are of the form $[1, \ldots, \widehat{i}, \ldots, k| \gamma]$ for $k \in [d], i \in [k]$ and $\gamma \in I(2d-k,k-1)$. Call the coefficient of this minor in the above $b(i,\gamma)$. Then
\[b(i,\gamma) = \sum_{1 \le j \le k} \sum_{\substack{\beta \in I(2d-k,k) \\ \gamma = (\beta_1, \ldots, \widehat{\beta}_j, \ldots, \beta_{k})}} c_{\beta}(-1)^{i+\beta_j} a_{i,\beta_j}^{(l)}. \]
The number of sequences $\beta$ considered by the inner sum is at most $O(d^2)$, and hence $b(i,\gamma)$ can be computed with $O(d^3)$ additions and multiplications. We can thus compute
\begin{equation}\label{minor_res}
\frac{\partial P}{\partial x_l} = \sum_{k=1}^{d} \sum_{i=1}^{k}\sum_{\gamma \in I(2d-k,k-1)} b(i,\gamma)[1, \ldots, \widehat{i}, \ldots, k|\gamma]
\end{equation}
with $d^4\sum_{k=1}^d |I(2d-k,k-1)| \le \varphi^{2d} \poly(d)$ arithmetic operations. Note that this expresses $\frac{\partial P}{\partial x_l}$ as a linear combination of minors that are not necessarily maximal. We now fix this.

 We first claim that for all $i \in [k]$ and $\beta \in I(2d-k,k-1)$,
 \[[1, \ldots, \widehat{i}, \ldots, k|\beta] = \sum_{J \subseteq [k-1], |J| = k-i} [e(J) + (1, \ldots, k-1) | \beta]\]
 where $e(J)$ is the indicator vector of the set $J$. This holds since when $J = \{i, \ldots, k-1\}$, $e(J) + (1, \ldots, k-1) = (1, \ldots, \widehat{i}, \ldots, k)$, and for all other $J$,  $e(J) + (1, \ldots, k-1)$ will have a repeated value and hence $[ e(J) + (1, \ldots, k-1) | \beta] = 0$. 
 
 Given this claim, it follows from \cite[Lemma 2.1(a)]{conca} that
 \begin{equation*}
 [1, \ldots, \widehat{i}, \ldots, k|\beta] = \sum_{J \subseteq [k-1], |J| = k-i} [\beta + e(J)],
 \end{equation*}
and so letting $Q_k$ be the degree-$k$ part of Equation \ref{minor_res},
 \[Q_k = \sum_{i=1}^{k+1}\sum_{\beta \in I(2d-k-1,k)} b(i,\beta) \sum_{J \subseteq [k], |J| = k+1-i} [\beta + e(J)].\]
We now show how to efficiently compute the coefficients of the maximal minors in this expression from the already computed $b(i,\gamma)$'s.

Let $0 \le k \le d-1$ be fixed.  For $\beta \in I(2d-k-1,k)$ and integers $i,j$ where $0 \le i \le j \le k$, let $D(\beta,i,j,k) \subseteq \{0,1\}^k$ be the set of binary vectors of length $k$ containing exactly $i$ ones, whose last $k-j$ entries are zero, and whose summation with $\beta$ is strictly increasing everywhere except possibly at positions $j$ and $j+1$ (that is, we may have $w_j+\beta_j = w_{j+1}+\beta_{j+1}$). Define 
\[A^k(i,j) \coloneqq \sum_{\beta\in I(2d-k-1,k)} b(k+1-i,\beta) \sum_{w \in D(\beta,i,j,k)} [\beta + w].\]
Note that $\sum_{i=0}^{k} A^k(i,k) = Q_k$, so it suffices to show how to compute $A^k(i,j)$ for all $i,j$. We do this with a dynamic program. When we store $A^k(i,j)$ we will store all coefficients of maximal minors arising in the above definition, even though such a minor might contain a repeated column and hence equal zero. The minors arising in this definition are specified by sequences of length $k$ with maximum value $2d-k$ that are strictly increasing everywhere but possibly at one position. Hence the number of such sequences is at most $k \binom{2d-k}{k}$. 

For the base cases, we have
\begin{align*}
A^k(0,j) &= \sum_{\beta \in I(2d-k-1,k)} b(k+1,\beta) [\beta],\\
A^k(i,i) &=  \sum_{\beta \in I(2d-k-1,k)} b({k+1-i,\beta}) [\beta + e(\{1, \ldots, i\})].
\end{align*}
Now suppose we have computed $A^k(i,j-1)$ and $A^k(i-1,j-1)$. Then
\begin{align*}
A^k(i,j) &=  \sum_{\beta \in I(2d-k-1,k)} b(k+1-i,\beta) \left (\sum_{\substack{w \in B(\beta,i,j,k), \\ w_j=0}} [\beta + w] + \sum_{\substack{w \in D(\beta,i,j,k),\\ w_j =1}} [\beta + w]\right ) \\
&= \sum_{\beta \in I(2d-k-1,k)} b(k+1-i,\beta) \sum_{\substack{w \in D(\beta,i,j-1,k), \\ \beta +w \text{ is strictly increasing}}} [\beta + w]  \\
&\qquad +\sum_{\beta \in I(2d-k-1,k)} b(k+1-i,\beta) \sum_{w \in D(\beta,i-1,j-1,k)} [\beta + w + e(\{j\})] .
\end{align*}
The first part of the sum can be computed from $A^k(i,j-1)$ by setting the coefficient of any maximal minor with a repeated column equal zero, and the second sum can be computed from $A^k(i-1,j-1)$ by setting the coefficient of $[\beta]$ to that of $[\beta - e(\{j\})]$. Hence $A^k(i,j)$ can be computed with $O(k\binom{2d-k}{k})$ arithmetic operations. It follows that we can represent $\frac{\partial P}{\partial x_l} = \sum_{i=0}^{d-1} Q_i$ in the space of maximal minors using $\varphi^{2d} \poly(d)$ arithmetic operations.
\end{proof}

With this we have the following analog of Theorem \ref{gendiff}. We omit the proof as it is almost exactly the same, we just work in the space of maximal minors rather than minors, using Lemma \ref{cat_ldiff} to differentiate instead of Lemma \ref{linear_diff_gen}.

\hankeldiff

\begin{cor}\label{hankel_cycle}
Let $G$ be a graph on $n$ vertices. We can decide in $\varphi^{2d} \poly(n)$ time if $G$ contains a simple cycle of length $d$.
\end{cor}
\begin{proof}
Let $V \in \mathbb{Q}^{d \times n}$ be the Vandermonde matrix with $V_{i,j}=j^i$, and  $X = V \cdot \diag(x_1, \ldots, x_n)\cdot V^\mathrm{T}$. By the argument of \ref{test_c}, $\langle \det X , \tr(A_G)^d \rangle \neq 0$ if and only if $G$ contains a simple cycle of length $d$. Note that the $(i,j)$th entry in $X$ equals $\sum_{k=1}^n k^{i+j} x_k$, and therefore $X$ is Hankel. We conclude by applying Theorem \ref{hankeldiff} to compute $\langle \det X, \tr(A_G)^d \rangle$, as $\tr(A_G^d)$ can be computed by a skew circuit of size $\poly(n)$.
\end{proof}
\begin{remark}
This algorithm extends to detecting subgraphs of bounded pathwidth on $d$ vertices by using the construction of the subgraph generating polynomial given in \cite[Appendix B]{extensor}.
\end{remark}
\section{Applications}
In this section we give our applications of Theorems 1 and 2.

\begin{cor}\label{pit}
Given matrices $A_1, \ldots, A_n \in \mathbb{Q}^{d \times d}$, we can decide if their span contains an invertible matrix in time $4^d \poly(N)$, where $N$ denotes the size of the input.
\end{cor}
\begin{proof}
Let $X = \sum_{i=1}^n x_iA_i$. First note that $\lin(A_1, \ldots, A_n)$ contains an invertible matrix if and only if $\det X \not \equiv 0$. Writing $\det X = \sum_{\alpha \in \mathbb{N}_d^n} c_\alpha x^\alpha$ for some coefficients $c_\alpha$ (at least one of which will be nonzero iff the answer is ``yes''), observe that $\langle \det X, \det X \rangle =  \sum_\alpha c_\alpha^2 \alpha!$. It follows that $\lin(A_1, \ldots, A_n)$ contains an invertible matrix if and only if this quantity is nonzero.

It is shown in \cite{mahajan1997combinatorial} that $\det_d$ can be expressed as a skew circuit of size $O(d^4)$, and the construction of this circuit is linear in the output size. Hence we can construct a circuit for $\det X$ by replacing the input variable $x_{ij}$ in this circuit with the $(i,j)$th entry of $X$. The theorem follows by applying Theorem \ref{gendiff} to the matrix $X$ and this circuit, noting that all numbers have bit-length $\poly(N)$ throughout the algorithm.
\end{proof}

\begin{cor}\label{matroid1}
Suppose we are given a matrix $A \in \mathbb{Q}^{km \times kn}$, where $n \ge m$, representing a matroid $M$ with groundset $[kn]$, and a partition $\pi$ of $[kn]$ into parts of size $k$. Then we can decide if the union of any $m$ parts in $\pi$ are independent in $M$ in time $4^{km} \poly (N)$, where $N$ is the size of the input.
\end{cor}
\begin{proof}
Let $g \coloneqq (\sum_{S \in \pi} \prod_{i \in S} x_i)^m$. It is easily seen that the square-free monomials appearing in $g$ correspond to unions of $m$ elements in $\pi$, and that $g$ can be computed  by a skew circuit of size $\poly(n)$. Next, let $X = A \cdot \diag(x_1, \ldots, x_n)\cdot A^\mathrm{T}$. By Cauchy-Binet,

\[\det X =  \sum_{S \in \text{Bases}(M)} \det(B_S)^2 \prod_{i \in S}x_i,\]
Note that the same monomial appears in the expansion of $g$ and $\det X$ exactly when there is such an independent set in $M$, and then since $g$ and $\det X$ have non-negative coefficients, $\langle \det X, g \rangle \neq 0$ if and only if an independent set in $M$ is the union of $m$ blocks in $\pi$. We conclude by applying Theorem \ref{gendiff}.
\end{proof}

Using the same trick as in \cite{marx} we can use Corollary \ref{matroid1} to solve the $k$-matroid intersection problem.
\begin{cor}[$k$-Matroid Intersection]\label{kmi}
Suppose we are given matrices $B_1, \ldots, B_k \in \mathbb{Q}^{m \times n}$ representing matroids $M_1, \ldots ,M_k$ with the common groundset $[n]$. We can decide if $M_1, \ldots, M_k$ share a common base in time $4^{km} \poly(N)$, where $N$ is the size of the input.
\end{cor}
\begin{proof}
Let $M = \bigoplus_{i=1}^k B_k$ be the direct sum of the input matrices. We first partition $[kn]$ into $n$ parts of size $k$ as follow: for $i \in [n]$, let $S_i \coloneqq \{i,i+n, i+2n, \ldots, i+kn\}$. If a union of $m$ of the blocks $S_1, \ldots, S_n$ are independent in the matroid represented by $M$, then $M_1, \ldots, M_k$ share a common base. Conversely, if these matroids share a common base, some union of the $S_i$'s are independent in the matroid represented by $M$. We conclude by applying Corollary \ref{matroid1} to the matrix $M \in \mathbb{Q}^{km \times kn}$ and the partition $S_1, \ldots, S_n$.
\end{proof}

Finally, we have our applications of Theorem \ref{hankeldiff}. These follow immediately by a reduction given in \cite[Theorem 1]{esa} to the following ``square-free monomial detection'' algorithm.
\begin{cor} 
Let $g \in \mathbb{Q}[x_1, \ldots, x_n]_d$ be a homogeneous degree-$d$ polynomial with nonnegative coefficients, computed by a skew arithmetic circuit $C$. Given as input $C$, we can decide in  deterministic $\varphi^{2d}  |C|  \poly(n)$ time whether $g$ contains a degree-$d$ square-free monomial.
\end{cor}
\begin{proof}
Let $V \in \mathbb{Q}^{d \times n}$ be the Vandermonde matrix with $V_{i,j}=j^i$, and  $X = V \cdot \diag(x_1, \ldots, x_n)\cdot V^\mathrm{T}$. By Cauchy-Binet,
\[\det X =  \sum_{S \subseteq \binom{[n]}{d}} \det(B_S)^2 \prod_{i \in S}x_i.\]
Since any $d$ columns in $B$ are linearly independent, $\det(B_S)^2 > 0$ for all $S$. It follows that since $g$ has nonnegative coefficients, $\langle \det X, g \rangle \neq 0$ if and only if $g$ contains a square-free monomial. Note that the $(i,j)$th entry in $X$ equals $\sum_{k=1}^n k^{i+j} x_k$, and therefore $X$ is Hankel. The theorem follows by invoking Theorem \ref{hankeldiff}.
\end{proof}
Applying \cite[Theorem 1]{esa}, we have:
\apps

\section{The Bilinear complexity of apolar algebras}
In this section we study the complexity of multiplication in apolar algebras as a first step towards generalizing Theorems \ref{gendiff} and \ref{hankeldiff}. We will work over $\mathbb{C}$ rather than $\mathbb{R}$ for convenience.
\subsection{Algebraic preliminaries}
\subsubsection{Apolarity}
Let $\mathcal{R}^n \coloneqq \mathbb{C}[\partial_1, \ldots, \partial_n]$ be the ring of partial differential operators. Elements of this ring are just multivariate polynomials in the variables $\partial_1, \ldots, \partial_n$. For an $n$-tuple $\alpha \in \mathbb{N}^n$, we let $\partial^\alpha$ be the monomial $\partial_1^{\alpha_1} \cdots \partial_n^{\alpha_n}$, and let $|\alpha| = \sum_{i=1}^n \alpha_i$.  For $h \in \mathcal{R}$ and $f \in \mathcal{S}$, we denote by $h \circ f$ the result of applying the differential operator $h$ to $f$. For example,
\[
(3\cdot\partial_1 \partial_2 + \partial_1^2) \act x_1^2x_2 = 
3\cdot\partial_1 \partial_2 \act x_1^2x_2 + \partial_1^2 \act x_1^2x_2 = 
6x_1 + 2x_2.
\]
It is clear that when $h$ and $f$ are homogeneous of the same degree, $h \circ f$ is a scalar. In this case $ f(\partial_1, \ldots, \partial_n) \circ g =\langle f, g \rangle$, so computing $h \circ f$ is equivalent to computing the apolar inner product.

\begin{defn}
For $f \in \mathcal{S}_d^n$, we define $\ann(f)$ as the ideal of elements in $\mathcal{R}^n$ annihilating $f$ under differentiation. We define the \emph{apolar algebra} $\sa_f$ as the quotient $\mathcal{R}^n/\ann(f)$.
\end{defn}
In other words, $\sa_f$ is the ring of representatives of equivalence classes of differential operators subject to the equivalence relation $\sim$, where $h \sim h'$ if and only if $h \circ f = h' \circ f$. It follows that there is a vector space isomorphism $\mathcal{J}$ between $\sa_f$ and $\diff(f)$, sending $h \in \sa_f$ to $h \circ f$. In particular, $(\sa_f)_i \cong \diff(f)_{d-i}$, where we denote by $(\sa_f)_i$ the vector space of degree-$i$ elements in $\sa_f$. 

\begin{remark}\label{rmkq}
Multiplication in $\sa_f$ corresponds to differentiating by $f$: for $h_1,h_2 \in \sa_f$, $\mathcal{J}(h_1 \cdot h_2) = h_1 \circ (h_2 \circ f)$. It follows that Lemmas \ref{linear_diff_gen} and \ref{cat_ldiff} are algorithms for multiplication by $\partial_l$ in $\sa_{\det X}$, with respect to the spanning sets of $\sa_{\det X}$ given by the inverse images of the minors (or maximal minors) of $X$.
\end{remark}

\begin{example}\label{mon_ex}
Let $f = x_1 x_2 \cdots x_n$. Note that for $1 \le i \le n$, $\partial_i^2 \circ f = 0$, and so $\partial_i^2 \in \ann(f)$. Moreover, it is not hard to see that $\partial_1^2, \ldots, \partial_n^2$ generate $\ann(f)$. So the apolar algebra of $f$ equals $\sa_f = \mathbb{C}[\partial_1, \ldots, \partial_n]/(\partial_1^2, \ldots, \partial_n^2)$. This ring has as a basis the set of square-free monomials $\{\prod_{i \in S} \partial_i\}_{S \subseteq [n]}$, and the product of two basis elements is given by the rule
\[
\partial_S \cdot \partial_T = \begin{cases} \partial_{S \cup T} &\mbox{if } S \cap T = \emptyset ,\\
0 & \mbox{else.}\end{cases}
\]
\end{example}

\subsubsection{Bilinear complexity}
We give a brief primer on bilinear complexity. We refer to Chapter 14 of \cite{burgisser2013algebraic} for an in-depth treathment of this topic.

Let $U,V,W$ be finite dimensional complex vector spaces, and let $U \otimes V \otimes W$ be the vector space of three-tensors. An element of $U \otimes V \otimes W$ of the form $u \otimes v \otimes w$ is called \emph{simple}. The \emph{rank} of a tensor $T \in U \otimes V \otimes W$, denoted $\mathbf{R}(T)$, is the smallest $r$ such that $T$ can be expressed as a sum of $r$ simple tensors.

A  $\mathbb{C}$-algebra $A = (V,\phi)$ is a complex vector space $V$ with a multiplication operation defined by a bilinear map $\phi : V \times V \to V$. We say $A$ is \emph{associative} if $\phi(v_1,\phi(v_2,v_3)) = \phi(\phi(v_1,v_2),v_3)$ for all $v_1,v_2,v_3 \in V$, and \emph{unital} if there is an element $e \in V$ such that $\phi(e,v) = \phi(v,e) = v$ for all $v \in V$. We will only be interested in unital associative algebras.

Let $\{e_1, \ldots, e_n\}$ be a basis for $V$ and $\{e_1^*, \ldots, e_n^*\}$ be its dual basis. We can naturally identify $A$ with its \emph{structure tensor}
\[\sum_{i,j \in [n]} e_i \otimes e_j \otimes (e_i \cdot e_j) = \sum_{i,j,k \in [n]} e_k^*(\phi(e_i,e_j)) e_i \otimes e_j \otimes e_k \in V \otimes V \otimes V.\]
As an abuse of notation, we denote by $\mathbf{R}(A)$ the rank of the structure tensor of $A$. The algorithmic importance of this quantity is that its at most twice the minimum number of non-scalar multiplications needed to compute the product of two elements in $A$ \cite[Equation 14.8]{burgisser2013algebraic}. Ranks of algebras are a classic topic in algebraic complexity (see \cite[Chapter 17]{burgisser2013algebraic}, with the following being the most notorious example.

\begin{example}
Let $M_n = (\mathbb{C}^{n \times n}, \phi)$ be the algebra of $n \times n$ complex matrices, where $\phi$ is given by matrix multiplication. The vector space $\mathbb{C}^{n \times n}$ has as a basis the set of matrices $\{e_{ij}\}_{i,j \in [n]}$, where $e_{ij}$ is the matrix whose $(i,j)$th entry equals one and all other entries equal zero. The multiplication of two basis elements is given by the rule $e_{ij} \cdot e_{kl} = e_{il}$ if $j=k$, and $e_{ij} \cdot e_{kl}=0$ otherwise. Hence the structure tensor of $A$ is $\langle n, n, n \rangle \coloneqq\sum_{i,j,k \in [n]} e_{ij}\otimes e_{jk} \otimes e_{ik}$,  the \emph{matrix multiplication tensor}. The exponent of matrix multiplication is defined as $\omega \coloneqq \inf_{c}\{ \mathbf{R}(M_n) \le O(n^c) \}$.
\end{example}

\begin{example}[Fast subset convolution]\label{fsc}
Let $f =  x_1 \cdots x_n$. We claim that the problem of multiplying elements in $\sa_f$ is exactly that of computing the \emph{subset convolution}. Here the subset convolutions is defined for functions $\sigma, \tau : 2^{[n]} \to \mathbb{C}$ as the function $\sigma * \tau : 2^{[n]} \to \mathbb{C}$ such that
\[(\sigma * \tau) (S) = \sum_{U \subseteq S} \sigma(U)\tau(S - U).\]
The problem of computing $(\sigma * \tau)(S)$ for all $S$, given as input the $2^n$ values of $\tau$ and $\sigma$, has a handful of applications in exact algorithms \cite{bjorklund2007fourier,fominbook}. We now elaborate on the connection between subset convolution and $\sa_f$.

 Using the basis of square-free monomials as in Example \ref{mon_ex}, define the elements $a = \sum_{S \subseteq [n]} \sigma(S)\partial_S, b = \sum_{S \subseteq [n]} \tau(S) \partial_S$ of $\sa_f$. Then by the equation for multiplication in $\sa_f$ given in \ref{mon_ex}
\[a \cdot b = \sum_{S \subseteq [n]}(\sigma * \tau) (S) \partial_S,\]
so we can compute the subset convolution by computing $a \cdot b$ and reading off the coefficients of the result. It follows that the minimum number of non-scalar multiplications necessary to compute the subset convolution is at most $2 \mathbf{R}(\sa_f)$ (recalling \cite[Equation 14.8]{burgisser2013algebraic}). The structure tensor of $\sa_f$ is

\[\sum_{S,T \subseteq [n]} \partial_S \otimes \partial _T \otimes (\partial_S \cdot \partial_T) = \sum_{S,T \subseteq [n], S \cap T = \emptyset } \partial_S \otimes \partial _T \otimes \partial_{S \cup T}.\]
This expression shows that the rank of this tensor is at most the number of pairs of disjoint subsets of $[n]$, which equals $3^n$ (each element in $[n]$ can be assigned to one of two subsets, or to none). In \cite{bjorklund2007fourier} an algorithm for computing subset convolution is given that uses just $O(\binom{n+2}{2} 2^n)$ multiplications, thus showing that $\mathbf{R}(\sa_f) \le O(\binom{n+2}{2} 2^n)$. In fact, one can say slightly more: the rank of this tensor has been studied in algebraic complexity, and it is known that $3 \cdot 2^n - o(2^n)\le \mathbf{R}(\sa_f) \le (2n+1)2^n$ \cite[Proposition 7,9]{zuiddam2017note}.
\end{example}

Our next Theorem relates $\mathbf{R}(\sa_f)$ to Waring rank.

\begin{theorem}\label{waringrelax}
Let $f \in \mathcal{S}_d^n$ and let $\sa_f$ be its apolar algebra. Then 
\[\mathbf{R}(\sa_f) \le (3d+1)  \mathbf{R}_S(f).\]
\end{theorem}
\begin{proof}
Suppose that $f=\sum_{i=1}^r b_i \ell_i^d$, where $\ell_i = (a_{i,1} x_1 + \cdots + a_{i,n}x_n)$. Let $B$ be a monomial basis for $\sa_f$. Let
\[T \coloneqq \sum_{\partial^\alpha,\partial^\beta \in B} \partial^\alpha \otimes \partial^{\beta} \otimes (\partial^{\alpha +\beta}\circ f) \in \sa_f \otimes \sa_f \otimes \diff(f).\]
First note that by the Apolarity lemma \cite[Lemma 1.15(i)]{ik},
\[\partial^{\alpha+\beta} \circ f = \frac{d!}{(d-|\alpha| - |\beta|)!}\sum_{i=1}^r c_i a_{i,1}^{\alpha_1+\beta_1}\cdots a_{i,n}^{\alpha_n + \beta_n} \ell_i^{d-|\alpha|-|\beta|}\]
and hence for an indeterminate $\varepsilon$,
\[T\varepsilon^d + \sum_{\substack{0 \le i \le 3d \\ i \neq d}} T_i \varepsilon^i = \sum_{i=1}^r \left (\sum_{\partial^\alpha \in B}\partial^{\alpha} a_{i,1}^{\alpha_1} \cdots a_{i,n}^{\alpha_n}\varepsilon^{|\alpha|} \right )\otimes \left (\sum_{\partial^\beta \in B}\partial^\beta a_{i,1}^{\beta_1} \cdots a_{i,n}^{\beta_n}\varepsilon^{|\beta|} \right ) \otimes \left ( \sum_{j=0}^d \frac{c_i d!}{(d-j)!} \ell_k^{d-j} \varepsilon^{d-j} \right )\]
since if $|\alpha| + |\beta| + (d-j) = d$, then $|\alpha| + |\beta| = j$. Here the $T_i$'s are ``junk'' tensors we'd like to get rid of. We do this with an interpolation trick. Let $\{\varepsilon_i\}_{0 \le i \le 3d}$ be elements of $\mathbb{C}$ that are distinct and nonzero, and let $\{c_i\}$ be the solution to the Vandermonde system
\[
 \sum_{i=0}^{3d} \varepsilon_i^j c_i = \begin{cases} 1, &\mbox{ } j=d ,\\
0, & \mbox{ } j\in \{0, \ldots, d-1,d+1,\ldots, 3d\}.\end{cases}
\]
Then $\sum_{i=0}^{3d} c_i ( T \varepsilon_i^d + \sum_{j \neq d} T_j \varepsilon_i^j) = T$, and hence $\mathbf{R}(T) \le (3d+1)\mathbf{R}_S(f)$. 

Finally, we claim that $T$ is isomorphic to the structure tensor of $\sa_f$. This follows by applying the vector space isomorphism between $\diff(f)$ and $\sa_f$ sending $h \circ f$ to $h$, which sends $\partial^{\alpha + \beta} \circ f$ to $\partial^\alpha \partial^\beta$.
\end{proof}

We also have the following simple lower bound:
\begin{prop}
$\mathbf{R}(\sa_f) \ge \dim \sa_f = \dim \diff(f)$.
\end{prop}
\begin{proof}
As $\sa_f$ is unital, the Proposition follows (see e.g.~\cite[Section 2.1]{zuiddam2017note}).
\end{proof}

The algorithmic relevance of $\mathbf{R}(\sa_f)$ to the computing apolar inner product is given explicitly by the following proposition.
\begin{prop}\label{multbd}
Fix $f \in \mathcal{S}_d^n$, and let $g \in \mathcal{S}_d^n$ be given as an arithmetic circuit $C$. Then we can compute $\langle f, g \rangle$ using $O(\mathbf{R}(\sa_f) |C|)$ non-scalar multiplications.
\end{prop}
\begin{proof}
Let $(\sa_f)_1$ have the basis $\partial_1, \ldots, \partial_k$ for some $k \le n$, and let $(\sa_f)_d$ have the basis $q$. Let $h = g(\partial_1, \ldots, \partial_n)$. Then the result of evaluating $h$ over $\sa_f$ equals $h \bmod Ann(f) = \frac{\langle f, g \rangle q}{\langle f, q \rangle}$. So, our algorithm will evaluate $C$ over $\sa_f$, obtaining $c \cdot q$ for some $c \in \mathbb{C}$. We then return $c\langle f, q \rangle$. Note that $\langle f, q \rangle$ does not depend on the input $g$.

To evaluate $h$, we first replace the input gates $x_i$ in $C$ by zero if $i > k$, and $\partial_i$ otherwise. We then evaluate $C$ inductively over $\sa_f$. At each gate we store an element of $\sa_f$, which can be encoded by a vector of length $\dim \sa_f$. At addition gates we simply sum the two inputs, which is done with $\dim \sa_f \le \mathbf{R}(\sa_f)$ additions, where the inequality follows by Proposition \ref{multbd}. Multiplication gates can be computed with at most $2\mathbf{R}(\sa_f)$ non-scalar operations by \cite[Equation 14.8]{burgisser2013algebraic}.
\end{proof}
\subsection{The bilinear complexity of $\sa_{\det_n}$}

\begin{theorem}\label{det_rank}
$\mathbf{R}(\sa_{\det_n}) \le O\left (n2^{\omega n} \right )$.
\end{theorem}
\begin{proof}
We first give a basis for $\sa_{\det_n}$ and describe how multiplication behaves with respect to this basis. This tells us what the structure tensor of the apolar algebra is. We then show how to obtain this tensor from $4n+1$ copies of $\langle 2^n, 2^n, 2^n \rangle$. We do this somewhat indirectly, using the fact that complex Clifford algebras are isomorphic to matrix algebras \cite[Chapter 13]{porteous_1981}. We assume that $n$ is even for ease of exposition.

We claim that the set of monomials of the form $(I | J) \coloneqq \partial_{I_1,J_1} \cdots \partial_{I_k,J_k}$, where $I,J \in I(n,k)$ and $0 \le k \le n$, are a basis for $\sa_{\det_n}$. This follows from the fact that there are $\binom{2n}{n}$ such monomials, $\dim \diff(\det_n) = \binom{2n}{n}$, and the polynomials of the form $(I | J ) \circ \det_n$ are linearly independent. The latter claim can be seen by noting that if $(I | J) \neq (I' | J')$, $(I | J ) \circ \det_n$ and $(I' | J' ) \circ \det_n$ have disjoint sets of monomials appearing in their expansion.

Next we claim that the product of two basis elements $(I | J)$ and $(I' | J')$ is given by the rule
\[
(I | J) \cdot (I' | J') = \begin{cases} 0 \mbox{ if } I \cap I' \neq \emptyset \text{ or } J \cap J' \neq \emptyset,\\
\sgn(I,I')\sgn(J,J')(I \cup I' | J \cup J')  \mbox{ else}\end{cases}
\]
 where $\sgn(I,I')$ denotes the sign of the permutation that brings the sequence $I_1, \ldots, I_k, I'_1, \ldots, I_{k'}$ into increasing order, and $I \cup I'$ denotes the resulting sorted sequence. Indeed, if $I \cap I' \neq \emptyset$, then $(I | J)(I'  | J')$ is divisible by the product of two variables that have the same first (row) index. But then $(I | J)(I'  | J') \circ \det_n = 0$, since all monomials in the determinant have different row indices. The second case follows from the fact that for $I,J \in I(n,k)$ and $\tau \in \mathfrak{S}_k$, $(I|J) \circ \det_n = \sgn \tau^{-1} \cdot (\tau(I) | J) \circ \det_n$, which follows from the Leibniz formula for the determinant. Therefore the structure tensor of $\sa_{\det_n}$ equals
\[T = \sum_{\substack{I,J,I',J' \subseteq [n]\\ |I| = |J|, |I'| = |J'| \\ I \cap I' = J \cap J' = \emptyset}} \sgn(I,I')\sgn(J,J') (I|J) \otimes (I' | J') \otimes (I \cup I' | J \cup J').\]

Let $CL_n = (V_n, \cdot)$ be the Clifford algebra of a nondegenerate quadratic form on $\mathbb{C}^n$ (see \cite[Chapter 13]{porteous_1981} for background). Concretely, $V_n$ has the basis $X_U$ for $U \subseteq [n]$, and the product of two basis elements is given by the rule
\[X_U \cdot X_{U'} = \sgn(U,U') X_{U \Delta U'}\]
where $\Delta$ denotes the symmetric difference of sets. Here $\sgn(U,U')$ is the sign of the permutation that brings $U,U'$ into nondecreasing order, leaving the relative order of any repeated elements unchanged. So the structure tensor of $CL_n$ is
\[T_n \coloneqq \sum_{U,U' \subseteq [n]} \sgn(U,U') X_U \otimes X_{U'} \otimes X_{U \Delta U'}.\]
Since $CL_n$ is isomorphic to the algebra of $2^{n/2} \times 2^{n/2}$ matrices \cite[Section 3]{nlab:clifford_algebra}, $\mathbf{R}(T_n) \le O(2^{\omega n/2})$. Thus by submultiplicativity of tensor rank under the tensor product,
\[T_n \otimes T_n = \sum_{U,V,U',V' \subseteq [n]} \sgn(U,U') \sgn(V,V')(X_U \otimes X_V) \otimes (X_{U'}\otimes X_{V'}) \otimes (X_{U \Delta U'} \otimes X_{V \Delta V'})\]
has rank at most $O(2^{\omega n})$. Now define the linear transformations $M,M' : V_n \otimes V_n \to \sa_{\det_n}(\varepsilon)$ given by $M(X_U \otimes X_V) = (U | V) \varepsilon^{|U| + |V|}$, and $M'(X_U \otimes X_V) = (U | V)\varepsilon^{-|U|-|V|}$, where $\varepsilon$ is some indeterminate. Applying $M$ to the first two factors of the above tensor and $M'$ to the third factor,
\[(M,M,M') \cdot (T_n \otimes T_n) = T + \sum_{i=1}^{4n} \varepsilon^{i} H_i,\]
for some ``junk'' tensors $H_i$, since $|U| + |U'| = |U \Delta U'|$ if and only if $U$ and $U'$ are disjoint. Applying the interpolation trick as in Theorem \ref{waringrelax}, it follows that $\mathbf{R}(\sa_{\det_n}) \le O ((4n+1)2^{\omega n})$.
\end{proof}
\begin{remark}
In fact, the above proof shows that the border rank of the structural tensor of $\sa_{\det_n}$ is at most $O(2^{\omega n})$.
\end{remark}

\section{Acknowledgments}
We would like to thank several anonymous reviewers for their comments on an earlier draft of this paper.

\bibliographystyle{alpha}
\bibliography{refs}

\newcommand{\etalchar}[1]{$^{#1}$}
\begin{thebibliography}{AOGSS17}

\bibitem[AG17]{anari2017generalization}
Nima Anari and Shayan~Oveis Gharan.
\newblock A generalization of permanent inequalities and applications in
  counting and optimization.
\newblock In {\em Proceedings of the 49th Annual ACM SIGACT Symposium on Theory
  of Computing}, pages 384--396. ACM, 2017.

\bibitem[AGR16]{anari2016monte}
Nima Anari, Shayan~Oveis Gharan, and Alireza Rezaei.
\newblock Monte {C}arlo {M}arkov chain algorithms for sampling strongly
  {R}ayleigh distributions and determinantal point processes.
\newblock In {\em Conference on Learning Theory}, pages 103--115, 2016.

\bibitem[AGV18]{anari2018log}
Nima Anari, Shayan~Oveis Gharan, and Cynthia Vinzant.
\newblock Log-concave polynomials, entropy, and a deterministic approximation
  algorithm for counting bases of matroids.
\newblock In {\em 2018 IEEE 59th Annual Symposium on Foundations of Computer
  Science (FOCS)}, pages 35--46. IEEE, 2018.

\bibitem[AHK18]{adiprasito2018hodge}
Karim Adiprasito, June Huh, and Eric Katz.
\newblock Hodge theory for combinatorial geometries.
\newblock {\em Annals of Mathematics}, 188(2):381--452, 2018.

\bibitem[AOGSS17]{anari2016nash}
Nima Anari, Shayan Oveis~Gharan, Amin Saberi, and Mohit Singh.
\newblock Nash social welfare, matrix permanent, and stable polynomials.
\newblock In {\em 8th Innovations in Theoretical Computer Science Conference
  (ITCS 2017)}. Schloss Dagstuhl-Leibniz-Zentrum fuer Informatik, 2017.

\bibitem[AR02]{arvind2002approximation}
Vikraman Arvind and Venkatesh Raman.
\newblock Approximation algorithms for some parameterized counting problems.
\newblock In {\em International Symposium on Algorithms and Computation}, pages
  453--464. Springer, 2002.

\bibitem[AR19]{alper2019syzygies}
Jarod Alper and Rowan Rowlands.
\newblock Syzygies of the apolar ideals of the determinant and permanent.
\newblock {\em Journal of Algebraic Combinatorics}, pages 1--36, 2019.

\bibitem[Bar95]{barvinok1995new}
Alexander~I Barvinok.
\newblock New algorithms for lineark-matroid intersection and matroid k-parity
  problems.
\newblock {\em Mathematical Programming}, 69(1-3):449--470, 1995.

\bibitem[Bar96]{barvinok1996two}
Alexander~I Barvinok.
\newblock Two algorithmic results for the traveling salesman problem.
\newblock {\em Mathematics of Operations Research}, 21(1):65--84, 1996.

\bibitem[BCS13]{burgisser2013algebraic}
Peter B{\"u}rgisser, Michael Clausen, and Mohammad~A Shokrollahi.
\newblock {\em Algebraic complexity theory}, volume 315.
\newblock Springer Science \& Business Media, 2013.

\bibitem[BDH18]{extensor}
Cornelius Brand, Holger Dell, and Thore Husfeldt.
\newblock Extensor-coding.
\newblock In {\em Proceedings of the 50th Annual {ACM} {SIGACT} Symposium on
  Theory of Computing, {STOC} 2018, Los Angeles, CA, USA, June 25-29, 2018},
  pages 151--164, 2018.

\bibitem[BHKK07]{bjorklund2007fourier}
Andreas Bj{\"o}rklund, Thore Husfeldt, Petteri Kaski, and Mikko Koivisto.
\newblock Fourier meets {M\"o}bius: fast subset convolution.
\newblock In {\em Proceedings of the thirty-ninth annual ACM symposium on
  Theory of computing}, pages 67--74, 2007.

\bibitem[BIP19]{burgisser2019no}
Peter B{\"u}rgisser, Christian Ikenmeyer, and Greta Panova.
\newblock No occurrence obstructions in geometric complexity theory.
\newblock {\em Journal of the American Mathematical Society}, 32(1):163--193,
  2019.

\bibitem[Bj{\"{o}}10]{bjoerklund}
Andreas Bj{\"{o}}rklund.
\newblock Determinant sums for undirected hamiltonicity.
\newblock In {\em 51th Annual {IEEE} Symposium on Foundations of Computer
  Science, {FOCS} 2010, October 23-26, 2010, Las Vegas, Nevada, {USA}}, pages
  173--182, 2010.

\bibitem[Bra19]{esa}
Cornelius Brand.
\newblock Patching colors with tensors.
\newblock In {\em 27th Annual European Symposium on Algorithms, {ESA} 2019,
  September 09-11, 2019, Munich, Germany}.

\bibitem[BT20]{boij2020bound}
Mats Boij and Zach Teitler.
\newblock A bound for the {W}aring rank of the determinant via syzygies.
\newblock {\em Linear Algebra and its Applications}, 587:195--214, 2020.

\bibitem[CFK{\etalchar{+}}15]{fominbook}
Marek Cygan, Fedor~V. Fomin, Lukasz Kowalik, Daniel Lokshtanov, Daniel Marx,
  Marcin Pilipczuk, Michal Pilipczuk, and Saket Saurabh.
\newblock {\em Parameterized Algorithms}.
\newblock Springer, 2015.

\bibitem[CHI{\etalchar{+}}18]{chiantini2018polynomials}
Luca Chiantini, Jonathan~D Hauenstein, Christian Ikenmeyer, Joseph~M Landsberg,
  and Giorgio Ottaviani.
\newblock Polynomials and the exponent of matrix multiplication.
\newblock {\em Bulletin of the London Mathematical Society}, 50(3):369--389,
  2018.

\bibitem[Con98]{conca}
Aldo Conca.
\newblock Straightening law and powers of determinantal ideals of {H}ankel
  matrices.
\newblock {\em Advances in Mathematics}, 138(2):263--292, 1998.

\bibitem[DT15]{derksen2015lower}
Harm Derksen and Zach Teitler.
\newblock Lower bound for ranks of invariant forms.
\newblock {\em Journal of Pure and Applied Algebra}, 219(12):5429--5441, 2015.

\bibitem[Edm67]{edmonds1967systems}
Jack Edmonds.
\newblock Systems of distinct representatives and linear algebra.
\newblock {\em J. Res. Nat. Bur. Standards Sect. B 71}, (4):241--245, 1967.

\bibitem[FG04]{flum2004parameterized}
J{\"o}rg Flum and Martin Grohe.
\newblock The parameterized complexity of counting problems.
\newblock {\em SIAM Journal on Computing}, 33(4):892--922, 2004.

\bibitem[FLPS16]{fomin}
Fedor~V. Fomin, Daniel Lokshtanov, Fahad Panolan, and Saket Saurabh.
\newblock Efficient computation of representative families with applications in
  parameterized and exact algorithms.
\newblock {\em J. {ACM}}, 63(4):29:1--29:60, 2016.

\bibitem[GGOW19]{garg2019operator}
Ankit Garg, Leonid Gurvits, Rafael Oliveira, and Avi Wigderson.
\newblock Operator scaling: theory and applications.
\newblock {\em Foundations of Computational Mathematics}, pages 1--68, 2019.

\bibitem[Gly13]{glynn2013permanent}
David~G Glynn.
\newblock Permanent formulae from the {V}eronesean.
\newblock {\em Designs, codes and cryptography}, 68(1-3):39--47, 2013.

\bibitem[GRWZ18]{gutin}
Gregory~Z. Gutin, Felix Reidl, Magnus Wahlstr{\"{o}}m, and Meirav Zehavi.
\newblock Designing deterministic polynomial-space algorithms by color-coding
  multivariate polynomials.
\newblock {\em J. Comput. Syst. Sci.}, 95:69--85, 2018.

\bibitem[Gur03]{gurvits}
Leonid Gurvits.
\newblock Classical deterministic complexity of {Edmonds}' problem and quantum
  entanglement.
\newblock In {\em Proceedings of the Thirty-fifth Annual ACM Symposium on
  Theory of Computing}, STOC '03, pages 10--19, New York, NY, USA, 2003. ACM.

\bibitem[Gur05]{gurvits2005complexity}
Leonid Gurvits.
\newblock On the complexity of mixed discriminants and related problems.
\newblock In {\em International Symposium on Mathematical Foundations of
  Computer Science}, pages 447--458. Springer, 2005.

\bibitem[Gur06]{gurvits2006hyperbolic}
Leonid Gurvits.
\newblock Hyperbolic polynomials approach to {V}an der
  {W}aerden/{S}chrijver-{V}aliant like conjectures: sharper bounds, simpler
  proofs and algorithmic applications.
\newblock In {\em Proceedings of the thirty-eighth annual ACM symposium on
  Theory of computing}, pages 417--426. ACM, 2006.

\bibitem[Gur08]{gurvits2008van}
Leonid Gurvits.
\newblock Van der waerden/schrijver-valiant like conjectures and stable (aka
  hyperbolic) homogeneous polynomials: one theorem for all.
\newblock {\em The electronic journal of combinatorics}, 15(1):66, 2008.

\bibitem[IK99]{ik}
Anthony Iarrobino and Vassil Kanev.
\newblock {\em Power sums, Gorenstein algebras, and determinantal loci}.
\newblock Springer Science \& Business Media, 1999.

\bibitem[KI04]{kabanets2004derandomizing}
Valentine Kabanets and Russell Impagliazzo.
\newblock Derandomizing polynomial identity tests means proving circuit lower
  bounds.
\newblock {\em Computational Complexity}, 13(1-2):1--46, 2004.

\bibitem[Lan12]{landsberg2012tensors}
Joseph~M Landsberg.
\newblock Tensors: geometry and applications.
\newblock {\em Representation theory}, 381:402, 2012.

\bibitem[Lee16]{lee2016power}
Hwangrae Lee.
\newblock Power sum decompositions of elementary symmetric polynomials.
\newblock {\em Linear Algebra and its Applications}, 492:89--97, 2016.

\bibitem[LG14]{le2014powers}
Fran{\c{c}}ois Le~Gall.
\newblock Powers of tensors and fast matrix multiplication.
\newblock In {\em Proceedings of the 39th international symposium on symbolic
  and algebraic computation}, pages 296--303, 2014.

\bibitem[Lov79]{Lovasz1979determinants}
L{\'a}szl{\'o} Lov{\'a}sz.
\newblock On determinants, matchings, and random algorithms.
\newblock pages 565--574, 1979.

\bibitem[Mac94]{macaulay1994algebraic}
Francis~Sowerby Macaulay.
\newblock {\em The algebraic theory of modular systems}, volume~19.
\newblock Cambridge University Press, 1994.

\bibitem[Mar09]{marx}
D{\'{a}}niel Marx.
\newblock A parameterized view on matroid optimization problems.
\newblock {\em Theor. Comput. Sci.}, 410(44):4471--4479, 2009.

\bibitem[MV97]{mahajan1997combinatorial}
Meena Mahajan and V~Vinay.
\newblock A combinatorial algorithm for the determinant.
\newblock In {\em In Proceedings of the 8th Annual ACM-SIAM Symposium on
  Discrete Algorithms}. Citeseer, 1997.

\bibitem[{nLa}20]{nlab:clifford_algebra}
{nLab authors}.
\newblock {{C}}lifford algebra.
\newblock \url{http://ncatlab.org/nlab/show/Clifford%20algebra}, April 2020.
\newblock
  \href{http://ncatlab.org/nlab/revision/Clifford%20algebra/22}{Revision 22}.

\bibitem[Por81]{porteous_1981}
Ian~R. Porteous.
\newblock {\em Topological Geometry}.
\newblock Cambridge University Press, 1981.

\bibitem[Pra19]{waring}
Kevin Pratt.
\newblock Waring rank, parameterized and exact algorithms.
\newblock In {\em 60th {IEEE} Annual Symposium on Foundations of Computer
  Science, {FOCS} 2019, Baltimore, MD, USA, November 9-12, 2019}, 2019.

\bibitem[Sha15]{shafiei1}
Masoumeh~Sepideh Shafiei.
\newblock Apolarity for determinants and permanents of generic matrices.
\newblock {\em Journal of Commutative Algebra}, 7(1):89--123, 2015.

\bibitem[Syl52]{sylvester1970principles}
J.J. Sylvester.
\newblock On the principles of the calculus of forms.
\newblock {\em Cambridge and Dublin Mathematical Journal}, 7:52--97, 1852.

\bibitem[Tsu19]{tsur}
Dekel Tsur.
\newblock {Faster deterministic parameterized algorithm for k-Path}.
\newblock {\em Theoretical Computer Science}, 790:96–104, 2019.

\bibitem[Wil09]{williams09}
Ryan Williams.
\newblock Finding paths of length k in ${O}^*(2^k)$ time.
\newblock {\em Inf. Process. Lett.}, 109(6):315--318, 2009.

\bibitem[Zui17]{zuiddam2017note}
Jeroen Zuiddam.
\newblock A note on the gap between rank and border rank.
\newblock {\em Linear Algebra and its Applications}, 525:33--44, 2017.

\end{thebibliography}

\end{document}